\documentclass[10pt]{article}
\usepackage{fullpage}
\usepackage{authblk}
\usepackage{graphicx}
\usepackage{amsmath}
\usepackage{amsfonts}
\usepackage{amssymb}
\usepackage{amsthm}
\usepackage{mathrsfs}
\usepackage{bm}
\usepackage{mathtools}
\usepackage{enumitem}
\usepackage{algorithm2e}

\usepackage{optidef}

\usepackage{url}
\usepackage{hyperref}

\usepackage{tcolorbox}
\newtcolorbox{wbox}
{
	colback  = white,
}

\usepackage{caption}
\usepackage{subcaption}
\newtheorem{theorem}{Theorem}
\newtheorem{lemma}[theorem]{Lemma}

\newtheorem{example}[theorem]{Example}
\newtheorem{alg}[theorem]{Mechanism}

\newtheorem*{claim*}{Claim}
\newtheorem*{remark*}{Remark}

\newtheorem*{definition*}{Definition}

\makeatletter
\renewcommand\section{%
  \@startsection{section}{1}
                {\z@}%
                {-3.5ex \@plus -1ex \@minus -.2ex}%
                {2.3ex \@plus.2ex}%
                {\large\bfseries}
}
\renewcommand\subsection{%
  \@startsection{subsection}{2}
                {\z@}%
                {-3.25ex\@plus -1ex \@minus -.2ex}%
                {1sp}
                {\normalsize\bfseries}
}
\renewcommand\subsubsection{%
  \@startsection{subsubsection}{3}
                {\z@}%
                {-3.25ex\@plus -1ex \@minus -.2ex}%
                {1sp}
                {\normalfont\normalsize}
}
\makeatother

\title{{\Large\bf  Approximate Core Allocations for\\Multiple Partners Matching Games}
\footnote{Supported in part by the National Natural Science Foundation of China (Nos.\,12001507, 11971447 and 11871442) and Natural Science Foundation of Shandong (No.\,ZR2020QA024).}}

\author{Han Xiao\thanks{Corresponding author. Email: hxiao@ouc.edu.cn.} ,
Tianhang Lu, Qizhi Fang}

\affil{School of Mathematical Sciences\\Ocean University of China\\Qingdao, China}

\date{}
\begin{document}


\maketitle

\openup 1.2\jot


\begin{abstract}
The matching game is a cooperative game where the value of every coalition is the maximum revenue of players in the coalition can make by forming pairwise disjoint partners.
The multiple partners matching game generalizes the matching game by allowing each player to have more than one possibly repeated partner.
In this paper, we study profit-sharing in multiple partners matching games.
A central concept for profit-sharing is the core which consists of all possible ways of distributing the profit among individual players such that the grand coalition remains intact.
The core of multiple partners matching games may be empty \cite{DIN99};
even when the core is non-empty, the core membership problem is intractable in general \cite{BKPW18}.
Thus we study approximate core allocations upon which a coalition may be paid less than the profit it makes by seceding from the grand coalition.
We provide an LP-based mechanism guaranteeing that no coalition is paid less than $\frac{2}{3}$ times the profit it makes on its own.
We also show that $\frac{2}{3}$ is the best possible factor relative to the underlying LP-relaxation.
Our result generalizes the work of Vazirani \cite{Vazi21} from matching games to multiple partners matching games.

\hfill



\noindent\textbf{JEL Classifcation: } C71, C78.
\hfill

\hfill
\end{abstract}


\section{Introduction}

\subsection{Background and Motivation}

Consider a network of companies such that any connected pair of companies can enter deals with revenue and that every company can only fulfill a limited number of deals.
The scenario above, taken from \cite{KTZ21}, can be modeled as the multiple partners matching game \cite{Soto92, BKPW18} which generalizes the classic matching game \cite{DIN99, BKP12, KP03, KPT20}.
The matching game is a cooperative game where the value of every coalition is the maximum revenue of players in the coalition can make by forming pairwise disjoint partners.
The multiple partners matching game generalizes the matching game by allowing each player to have more than one possibly repeated partner.
The multiple partners setting makes the game model more applicable to the real world and more attractive to researchers \cite{SS71, Gran84, TM98, BKP12, BKPW18}.

The multiple partners matching game, also known as the $b$-matching game \cite{DIN99, KM20-IJCAI, KTZ21}, is defined as follows.
Throughout the paper, $G=(V,E)$ denotes an undirected graph with an \emph{integral vertex capacity function} $b$ and an \emph{edge weight function} $w$.
For any vertex $v\in V$, $\delta (v)$ dentoes the set of edges incident to $v$.
An integral vector $\boldsymbol{x}\in \mathbb{Z}^{E}$ is called a \emph{$b$-matching} of $G$ if it satisfies constraints \eqref{eq:capacity} and \eqref{eq:nonnegativity}
\begin{alignat}{3}
&&\quad \sum_{e\in \delta(v)} x_e&\leq b_v &\qquad \forall ~v &\in V,\label{eq:capacity}\\
&& \quad x_e &\geq 0 &\qquad \forall ~e &\in E,\label{eq:nonnegativity}
\end{alignat}
and a \emph{maximum $b$-matching} of $G$ if it further maximizes $\sum_{e\in E}w_e x_e$.
The \emph{multiple partners matching game} or \emph{$b$-matching game} on $G$, denoted by $\Gamma_G=(N,\gamma)$, is a cooperative game with player set $N=V$ and characteristic function $\gamma: 2^N\rightarrow \mathbb{R}$.
For any \emph{coalition} $S\subseteq N$, $\gamma(S)$ is defined to be the weight of maximum $b$-matchings in the induced subgraph $G[S]$.
Among all possible coalitions, the most important one is the \emph{grand coalition} $N$.
The worth of the game $\Gamma_G$ is determined by the grand coalition, which equals the weight of maximum $b$-matchings in $G$.
When the underlying graph is bipartite, the multiple partners matching game is called the \emph{multiple partners assignment game}.
When every vertex has capacity one, the multiple partners matching game becomes the classic \emph{matching game}.

A central question in cooperative game theory is to allocate the total profit generated through cooperation among individual players.
However, not all allocations are equally desirable. 
Emphases on different criteria lead to different allocation concepts.
The core, which addresses the issue of stability, i.e., any coalition will not be able to obtain more profits on its own by seceding from the grand coalition, is one of the most attractive allocation concepts.

An \emph{allocation} of the multiple partners matching game $\Gamma_G=(N,\gamma)$ is a vector $\boldsymbol{a}\in \mathbb{R}^N_{\geq 0}$ specifying how to allocate the worth of the game among individual players.
The \emph{core} of the game $\Gamma_G=(N,\gamma)$ is the set of allocations $\boldsymbol{a}\in \mathbb{R}^N_{\geq 0}$ satisfying
\begin{enumerate}
  \item[\textendash]  $\sum_{i\in N}a_i =\gamma (N)$,
  \item[\textendash]  $\sum_{i\in S}a_i \geq \gamma (S)$ for any $S\subseteq N$.
\end{enumerate}
Thus the core captures all possible ways of distributing the worth of the game among individual players such that the grand coalition remains intact.
Deng et al. \cite{DIN99} showed that the core of multiple partners assignment games is non-empty.
However, the core of multiple partners matching games is not guaranteed to be non-empty in general.
Bir\'{o} et al. \cite{BKPP19,BKPW18} showed that the core non-emptiness problem and the core membership problem for multiple partners matching games are solvable in polynomial time if $b\leq \textbf{2}$;
however, they are co-NP-hard even for $b = \boldsymbol{3}$.

\subsection{Our contributions}
In this paper, we study approximate core allocations for multiple partners matching games.
An allocation $\boldsymbol{a}\in \mathbb{R}^N_{\geq 0}$ lies in the \emph{$\alpha$-approximate core} ($0<\alpha\leq 1$) of the game $\Gamma_G=(N,\gamma)$ if
\begin{enumerate}
  \item[\textendash]  $\sum_{i\in N} a_i \leq \gamma (N)$
  \item[\textendash]  $\sum_{i\in S} a_i \geq \alpha\cdot\gamma (S)$ for any $S\subseteq N$.
\end{enumerate}
Any allocation in the $\alpha$-approximate core implies that,
under the given allocation, no coalition will be paid less than $\alpha$ times the profit it makes on its own, or equivalently, will gain more than a $\frac{1}{\alpha}$ factor by seceding from the grand coalition.
Thus, it is interesting to find the largest $\alpha$ such that the $\alpha$-approximate core is non-empty, especially when the core is empty.
A mechanism computing allocations in the $\alpha$-approximate core with the largest $\alpha$ is even more desirable.

We propose an LP-based mechanism that always computes an allocation in the $(1 - \frac{1}{l_{\min} \cdot b_{\min}})$-approximate core of the game $\Gamma_G$,
where $l_{\min}$ is the minimum length of odd cycles and $b_{\min}$ is the minimum capacity of vertices.
Consequently, $\frac{2}{3}$-approximate core of multiple partners matching games is always non-empty.
Moreover, $\frac{2}{3}$ is the best possible factor for multiple partners matching games as it coincides with the integrality gap of the underlying LP-relaxation.
Our result generalizes the mechanism of Vazirani \cite{Vazi21} from the matching setting to the $b$-matching setting.
Both mechanisms are LP-based and rely heavily on matching theory, LP duality theory, and their highly non-trivial interplay.
Compared to the mechanism of Vazirani, our mechanism is more efficient and achieves a better approximation ratio for $b$-matching instances.
Besides, we provide a higher point of view for interpreting the duality principle and the rounding technique behind the LP-based mechanism.

\subsection{Related work}

The matching game, introduced by Shapley and Shubik \cite{SS71}, is a special case of the multiple partners matching game.
It has attracted many research efforts and various allocation concepts have been studied within different settings \cite{SR94, FKFH98, DIN99, KP03, KPT20, XF21}.
Recently, Vazirani \cite{Vazi21} studied the approximate core for matching games and achieved the best possible approximation ratio.

The multiple partners matching game,
on one hand generalizes the matching game by allowing each player to have more than one possibly repeated partner,
on the other hand serves as a special case of the linear production game \cite{Owen75,DIN99, KM20-IJCAI}.
Sotomayor \cite{Soto92} surveyed multiple partners assignment games.
Deng et al. \cite{DIN99} showed that the core of multiple partners assignment games is always non-empty.
Bir\'{o} et al. \cite{BKPP19,BKPW18} showed that the core non-emptiness problem and the core membership problem for multiple partners matching games are solvable in polynomial time if $b\leq \textbf{2}$;
however, they are co-NP-hard even for $b = \boldsymbol{3}$.
In addition to the core, other properties and solution concepts have also been studied for multiple partners matching games.
Kumabe and Maehara \cite{KM20-IJCAI} provided a necessary and sufficient characterization for the convexity of multiple partners matching games.
K\"{o}nemann, Toth and Zhou \cite{KTZ21} discussed the complexity of nucleolus computation in multiple partners assignment games.

\section{Preliminaries}
\subsection{LP-relaxation for the maximum $b$-matching problem}
To study approximate core allocations for the multiple partners matching game $\Gamma_G=(N,\gamma)$,
we resort to the natural LP-relaxation for the maximum $b$-matching problem on $G=(V,E)$,
which can be formulated as LP \eqref{eq:core-primal}.
\begin{maxi!}
	{}{\sum_{ij\in E} w_{ij} x_{ij} \label{eq:core-primal.objective}}
	{\label{eq:core-primal}}{}
	\addConstraint { \sum_{ij\in \delta (i)} x_{ij} }{\leq b_{i} ,\quad }{\forall i\in V \label{eq:core-primal.c1}}
	\addConstraint { x_{ij} }{\geq 0 ,\quad }{\forall ij\in E . \label{eq:core-primal.c2}}
\end{maxi!}

Any feasible solution to LP \eqref{eq:core-primal} is called a \emph{fractional $b$-matching} of $G$.
Any optimal solution to LP \eqref{eq:core-primal} is called a \emph{maximum fractional $b$-matching} of $G$.
For any feasible solution $\boldsymbol{x}$ to LP \eqref{eq:core-primal} and any subgraph $H$ of $G$, we use $w(H,\boldsymbol{x})$ to denote $\sum_{e\in H} w_e x_e$.

Taking $y_i$ to be the dual variable for the constraint \eqref{eq:core-primal.c1} gives the dual LP of LP \eqref{eq:core-primal}.
\begin{mini!}
	{}{\sum_{i\in V} b_i y_i \label{eq:core-dual.objective}}
	{\label{eq:core-dual}}{}
	\addConstraint { y_i + y_j}{\geq w_{ij} ,\quad }{\forall ij\in E \label{eq:core-dual.c1}}
	\addConstraint { y_i }{\geq 0 ,\quad }{\forall i\in V . \label{eq:core-dual.c2}}
\end{mini!}

Any feasible solution to LP \eqref{eq:core-dual} is called a \emph{fractional $w$-cover} of $G$.
Any optimal solution to LP \eqref{eq:core-dual} is called a \emph{minimum fractional $w$-cover} of $G$.
For any feasible solution $\boldsymbol{y}$ to LP \eqref{eq:core-dual} and any subgraph $H$ of $G$, we use $b(H,\boldsymbol{y})$ to denote $\sum_{v\in H} b_v y_v$.

Notice that if LP \eqref{eq:core-primal} has an integral optimal solution $\boldsymbol{x}$, then the worth of the game $\Gamma_G$ equals the optimum of LP \eqref{eq:core-primal}.
Consequently, every dual optimal solution $\boldsymbol{y}$ defines a core allocation $\boldsymbol{a}\in \mathbb{R}^{N}$ for the game $\Gamma_G$ with $a_i = b_i y_i$ for $i\in N$.
This observation can be verified via the LP duality theory.
For the grand coalition $N$, we have
\begin{equation*}
\gamma (N) = \sum_{ij\in E} w_{ij} x_{ij} = \sum_{i\in V} b_i y_i=\sum_{i\in V} a_i.
\end{equation*}
For any coalition $S\subseteq N$ and any maximum $b$-matching $\boldsymbol{z}$ in the induced subgraph $G[S]$,
we have
\begin{equation*}
\gamma (S) = \sum_{ij\in G[S]} w_{ij} z_{ij} \leq \sum_{{ij} \in G[S]} (y_i + y_j) z_{ij} \leq \sum_{i\in S} b_i y_i =\sum_{i\in S} a_i.
\end{equation*}
Thus the vector $\boldsymbol{a}$ defined from dual optimal solution $\boldsymbol{y}$ is indeed a core allocation of the game $\Gamma_G$.

When $G$ is bipartite, it is well known that LP \eqref{eq:core-primal} always admits an integral optimal solution \cite{Schr86}.
Thus the core of $\Gamma_G$ is non-empty when $G$ is bipartite \cite{DIN99}.
However, the core of $\Gamma_G$ is not guaranteed to be non-empty in general.
Hence a natural question is to determine the largest $\alpha$ for the $\alpha$-approximate core of $\Gamma_G$ to be non-empty.
The work of Vazirani \cite{Vazi21} for matching games is built on the half-integrality of fractional matching polytopes.
It turns out that the fractional $b$-matching polytope is also half-integral,
which enables us to study approximate core allocations for multiple partners matching games.

\subsection{Half-integrality of fractional $b$-matching polytopes}

The polytope defined by the constraints \eqref{eq:core-primal.c1}-\eqref{eq:core-primal.c2} is called the \emph{fractional $b$-matching polytope} of $G$.
A vector $\boldsymbol{x}$ is called \emph{half-integral} if $2 \boldsymbol{x}$ is integral.
Balinski \cite{Bali65} showed that every vertex of the fractional $b$-matching polytope of $G$ is half-integral when $b \equiv \boldsymbol{1}$.
Based on the observation of Balinski \cite{Bali65}, it is easy to see that every vertex of the fractional $b$-matching polytope of $G$ is also half-integral when $b\geq \boldsymbol{1}$.

\begin{lemma}[Balinski \cite{Bali65}]
\label{thm:Balinski}
Any square non-decomposable $(0,1)$-matrix with at most two $1$'s in any column has determinant $0$, $\pm 1$ or $\pm 2$.
\end{lemma}

According to Lemma \ref{thm:Balinski}, the half-integrality of fractional $b$-matching polytopes follows directly.

\begin{lemma}\label{thm:half.integrality}
Every vertex of the fractional $b$-matching polytope is half-integral.
\end{lemma}

Lemma \ref{thm:half.integrality} implies that LP \eqref{eq:core-primal} always admits a half-integral optimal solution,
which is the starting point of our profit-sharing mechanism for multiple partners matching games.

\section{A mechanism for approximate core allocations}

In this section, we introduce an LP-based mechanism for multiple partners matching games which guarantees that no coalition is paid less than $\frac{2}{3}$ times the profit it makes on its own.
Our mechanism is built on the natural LP-relaxation for the maximum $b$-matching problem.
The general idea of our mechanism is to round a half-integral optimal solution to a feasible integral solution which is used to guide the construction of approximate core allocations from dual optimal solutions.
Our mechanism can be sketched as follows.
First, compute a half-integral optimal solution.
Then, modify the optimal solution without violating feasibility and optimality to make all non-integral edges form vertex-disjoint odd cycles.
Next, round the new optimal solution into an integral feasible solution such that the rounding loss is minimized in every odd cycle.
Finally, define an approximate core allocation from any dual optimal solution by referring to the rounding loss in every odd cycle.
In the following, we illustrate our mechanism with the multiple partners matching game $\Gamma_G=(N,\gamma)$.

\subsection{Computing a half-integral optimal $b$-matching}
\label{sec:half.integral.bmatching}

To compute an approximate core allocation for $\Gamma_G$, we have to obtain a half-integral optimal solution to LP \eqref{eq:core-primal}, i.e., a half-integral optimal $b$-matching in $G$.
To this end, we use the doubling idea of Nemhauser and Trotter \cite{NT75} to construct an auxiliary bipartite graph $\hat{G}=(\hat{V},\hat{E})$ from $G=(V,E)$.
Recall that $G$ has integral vertex capacities $b$ and edge capacities $w$.
We construct $\hat{G}$ with integral vertex capacities $\hat{b}$ and edge capacities $\hat{w}$ as follows.
For each vertex $i\in V$ with capacity $b_i$, $\hat{V}$ has two vertices $i'$ and $i''$ with capacities $\hat{b}_{i'}=\hat{b}_{i''}=b_i$.
For each edge ${ij} \in E$ with weight $w_{ij}$, $\hat{E}$ has two edges $i' j''$ and $i'' j'$ with weights $\hat{w}_{i' j''}=\hat{w}_{i'' j'}= w_{ij}$.
Since every cycle of length $k$ in $G$ is transformed into a cycle of length $2k$ in $\hat{G}$, $\hat{G}$ is bipartite.
Thus for bipartite graph $\hat{G}$, a maximum $b$-matching and a minimum $w$-cover can be computed in polynomial time \cite{Schr86}, say $\hat{\boldsymbol{x}}$ and $\hat{\boldsymbol{y}}$, respectively.
Next, we define $\boldsymbol{x}$ and $\boldsymbol{y}$ in \eqref{eq:half.integral.bmatching}.
\begin{equation}\label{eq:half.integral.bmatching}
x_{ij}=\frac{1}{2} (\hat{x}_{i' j''} + \hat{x}_{i'' j'}) \text{~~and~~} y_{i}= \frac{1}{2}(\hat{y}_{i'}+ \hat{y}_{i''}).
\end{equation}
It is easy to see that $\boldsymbol{x}$ and $\boldsymbol{y}$ obtained above are feasible to LPs \eqref{eq:core-primal} and \eqref{eq:core-dual}, respectively.
Moreover, $w(G,\boldsymbol{x})= b(G,\boldsymbol{y})$.
Thus $\boldsymbol{x}$ is a half-integral optimal $b$-matching and $\boldsymbol{y}$ is an optimal fractional $w$-cover.

\begin{lemma}\label{thm:half.integral.polytope}
$\boldsymbol{x}$ is a half-integral optimal $b$-matching and $\boldsymbol{y}$ is an optimal fractional $w$-cover.
\end{lemma}
\begin{proof}
We first show that $\boldsymbol{x}$ and $\boldsymbol{y}$ defined in \eqref{eq:half.integral.bmatching} are feasible to LPs \eqref{eq:core-primal} and \eqref{eq:core-dual}, respectively.
For any vertex $i\in V$, we have
\begin{equation*}
\sum_{ij\in \delta (i)} x_{ij}=\frac{1}{2}\sum_{i' j''\in \delta (i')} \hat{x}_{i' j''} + \frac{1}{2} \sum_{i'' j'\in \delta (i'')} \hat{x}_{i'' j'}\leq \frac{1}{2}(\hat{b}_{i'} + \hat{b}_{i''}) = b_{i},
\end{equation*}
which implies the feasibility of $\boldsymbol{x}$.
For any edge $ij\in E$, we have
\begin{equation*}
y_{i}+y_{j} = \frac{1}{2} (\hat{y}_{i'}+ \hat{y}_{j''}) + \frac{1}{2} (\hat{y}_{i''}+ \hat{y}_{j'}) \geq \frac{1}{2} (\hat{w}_{i' j''} + \hat{w}_{i'' j'}) = w_{ij},
\end{equation*}
which implies the feasibility of $\boldsymbol{y}$.

It remains to show that $w(G,\boldsymbol{x})= b(G,\boldsymbol{y})$.
On one hand,
\begin{equation}\label{eq:misc.1}
\sum_{ij\in E} w_{ij} x_{ij}=\frac{1}{2}\sum_{i' j'' \in \hat{E}} \hat{w}_{i' j''}\hat{x}_{i' j''} + \frac{1}{2}\sum_{i'' j' \in \hat{E}} \hat{w}_{i'' j'} \hat{x}_{i'' j'} = \frac{1}{2} \sum_{e\in \hat{E}} \hat{w}_{e} \hat{x}_e.
\end{equation}
On the other hand,
\begin{equation}\label{eq:misc.2}
\sum_{i \in V} b_{i} y_{i}= \frac{1}{2} \sum_{i'\in \hat{V}} \hat{b}_{i'} \hat{y}_{i'} + \frac{1}{2} \sum_{i''\in \hat{V}} \hat{b}_{i''} \hat{y}_{i''} = \frac{1}{2} \sum_{v\in \hat{V}} \hat{b}_{v} \hat{y}_{v}.
\end{equation}
By the LP duality theory, \eqref{eq:misc.1} and \eqref{eq:misc.2} imply $w(G,\boldsymbol{x})= b(G,\boldsymbol{y})$.
And the half-integrality of $\boldsymbol{x}$ follows from the integrality of $\hat{\boldsymbol{x}}$.
\end{proof}

Alternatively, a half-integral optimal $b$-matching can also be obtained by using the idea of Anstee \cite{Anst87},
where every fractional $b$-matching $\boldsymbol{x}$ is viewed as a symmetric matrix $(a_{i, j})_{\lvert V\rvert \times \lvert V\rvert}$ with $a_{i,j}=a_{j,i}=x_{ij}$.
To obtain the symmetric matrix corresponding to a half-integral optimal $b$-matching,
first solve a minimum cost flow problem defined from the maximum $b$-matching problem, which returns an asymmetric integral matrix,
then the matrix is manipulated to obtain a symmetric half-integral matrix corresponding to a half-integral optimal $b$-matching.
We refer to the work of Anstee \cite{Anst87} for a detailed explanation.

\subsection{Modifying the half-integral optimal $b$-matching}
Let $\boldsymbol{x}$ and $\boldsymbol{y}$ denote optimal solutions to LPs \eqref{eq:core-primal} and \eqref{eq:core-dual} obtained in Subsection \ref{sec:half.integral.bmatching}, respectively.
Clearly, $w(G,\boldsymbol{x})=b(G,\boldsymbol{y})$.
By the principle of complementary slackness, for any edge $ij\in E$,
\begin{equation}\label{eq:complementary-lackness.edge}
x_{ij} \cdot (y_i + y_j - w_{ij})=0,
\end{equation}
and for any vertex $i\in V$,
\begin{equation}\label{eq:complementary-lackness.vertex}
y_i \cdot (\sum_{ij\in \delta (i)} x_{ij} - b_i ) =0.
\end{equation}
Lemma \ref{thm:half.integral.polytope} guarantees that $\boldsymbol{x}$ is half-integral.

Now we focus on the subgraph of $G$ induced by \emph{non-integral} edges of $\boldsymbol{x}$,
which is crucial for our mechanism design.
Let $G'=(V',E')$ denote the subgraph of $G$ induced by non-integral edges of $\boldsymbol{x}$, i.e., $e\in E'$ if and only if $x_{e}$ is not integral.
Let $T=e_1 e_2 \ldots e_l$ with $e_i=v_i v_{i+1}$ be a trail (walk with repeated vertices allowed) in $G'$.
Perform the following operations on $\boldsymbol{x}$ along $T$ for $i=1,2,\ldots, l$,
\begin{equation}\label{eq:operation-integrality}
x_{e_i}=\begin{cases}
		  x_{e_i}-\frac{1}{2} \quad &i \text{ is odd}, \\
          x_{e_i}+\frac{1}{2} \quad &i \text{ is even}. \\  
     \end{cases}
\end{equation}
As we shall see, for certain trail $T$ in $G'$, the operations \eqref{eq:operation-integrality} change neither the feasibility nor the optimality of $\boldsymbol{x}$ to LP \eqref{eq:core-primal}.

\begin{lemma}\label{thm:odd.vertex.elimination}
Let $T$ be a trail joining two distinct odd degree vertices in $G'$.
Let $\boldsymbol{x}'$ be obtained from $\boldsymbol{x}$ by performing the operations \eqref{eq:operation-integrality} along $T$.
Then $\boldsymbol{x}'$ is feasible to LP \eqref{eq:core-primal}.
Moreover,
\begin{equation}
w(T,\boldsymbol{x}')=w(T,\boldsymbol{x}).
\end{equation}
\end{lemma}
\begin{proof}
Observe that the operations \eqref{eq:operation-integrality} do not affect inequalities in \eqref{eq:core-primal.c1} except at the end vertices of $T$.
Since the two end vertices $v_{1}$ and $v_{l+1}$ of $T$ have odd degree in $G'$, it follows that  $0 < \sum_{e\in \delta (v_1)} x_{e}<b_{v_1}$ and $0 < \sum_{e\in \delta (v_{l+1})} x_{e}<b_{v_{l+1}}$.
Thus $\boldsymbol{x}'$ is feasible to LP \eqref{eq:core-primal}.
By the principle of complementary slackness, \eqref{eq:complementary-lackness.edge} 
implies that $y_{v_i} +y_{v_{i+1}} = w_{e_i}$ as $x_{e_i}>0$ for $i=1,\ldots, l$, and \eqref{eq:complementary-lackness.vertex} implies that $y_{v_1}=y_{v_{l+1}}=0$ as $\sum_{e\in \delta (v_1)} x_{e}<b_{v_1}$ and $\sum_{e\in \delta (v_{l+1})} x_{e}<b_{v_{l+1}}$.
It follows that
\begin{equation}\label{eq:odd.vertex.elimination}
\begin{split}
      w(T, \boldsymbol{x}')
      = & ~ \sum_{i=1}^{l} (y_{v_i}+y_{v_{i+1}}) x'_{e_i} \\
      = & ~ \sum_{i=2}^{l} y_{v_i} (x'_{e_{i-1}}+x'_{e_{i}})\\
      = & ~ \sum_{i=2}^{l} y_{v_i} (x_{e_{i-1}}+x_{e_{i}}) \\
      = & ~ \sum_{i=1}^{l} (y_{v_i}+y_{v_{i+1}}) x_{e_i} \\
      = & ~ w(T,\boldsymbol{x}).
\end{split}
\end{equation}
\end{proof}

\begin{lemma}\label{thm:circuit.operation}
Let $T$ be a closed trail in $G'$.
For any vertex $v$ in $T$, let $\boldsymbol{x}'$ be obtained from $\boldsymbol{x}$ by performing the operations \eqref{eq:operation-integrality} along $T$ starting with $v$.
Then $\boldsymbol{x}'$ is feasible to LP \eqref{eq:core-primal}.
Moreover,
\begin{enumerate}[label={\emph{($\roman*$)}}]
	\item \label{itm:even.circuit.operation} if $l$ is even, then $w(T,\boldsymbol{x}')=w(T,\boldsymbol{x})$;
	\item \label{itm:odd.circuit.operation} if $l$ is odd, then $w(T,\boldsymbol{x}')=w(T,\boldsymbol{x})-y_{v}$.
\end{enumerate}
\end{lemma}
\begin{proof}
Without loss of generality, assume that we perform the operations \eqref{eq:operation-integrality} from $v_1$ along $T$.
Notice that the operations \eqref{eq:operation-integrality} do not affect inequalities in \eqref{eq:core-primal.c1} except for $v_1$.
Since $\sum_{e\in \delta (v_1)} x_e \geq \sum_{e\in \delta (v_1)} x'_e \geq \sum_{e\in \delta (v_1)} x_e -1$,
$\boldsymbol{x}'$ is feasible to LP \eqref{eq:core-primal}.
Moreover,
\begin{equation}\label{eq:circuit.operation}
	\begin{split}
		  w(T,\boldsymbol{x}') 
		  = & ~ \sum_{i=1}^{l} w_{e_i} x'_{e_i}\\
		  = & ~ (y_{v_1}+y_{v_{2}}) x'_{e_1} + (y_{v_2}+y_{v_{3}}) x'_{e_2} + \ldots + (y_{v_l}+y_{v_{1}}) x'_{e_l} \\
		  = & ~ y_{v_1} (x'_{e_l}+ x'_{e_1}) + y_{v_2} (x'_{e_1}+ x'_{e_2}) + \ldots + y_{v_l} (x'_{e_{l-1}}+ x'_{e_l}).
	\end{split}
\end{equation}
We proceed by distinguishing the parity of $l$.

\emph{\ref{itm:even.circuit.operation}}
If $l$ is even, then \eqref{eq:circuit.operation} implies that
\begin{equation}\label{eq:even.circuit.operation}
	\begin{split}
		  w(T,\boldsymbol{x}') 
		  = & ~ y_{v_1} (x_{e_l}+ x_{e_1}) + y_{v_2} (x_{e_1}+ x_{e_2}) + \ldots + y_{v_l} (x_{e_{l-1}}+ x_{e_l}) \\
		  = & ~ (y_{v_1}+y_{v_{2}}) x_{e_1} + (y_{v_2}+y_{v_{3}}) x_{e_2} + \ldots + (y_{v_l}+y_{v_{1}}) x_{e_l} \\
		  = & ~ \sum_{i=1}^{l} w_{e_i} x_{e_i} \\
		  = & ~ w(T,\boldsymbol{x}).
	\end{split}
\end{equation}

\emph{\ref{itm:odd.circuit.operation}}
If $l$ is odd, then \eqref{eq:circuit.operation} implies that
\begin{equation}\label{eq:odd.circuit.operation}
	\begin{split}
		  w(T,\boldsymbol{x}') 
		  = & ~ y_{v_1} (x_{e_l}+ x_{e_1} - 1) + y_{v_2} (x_{e_1}+ x_{e_2}) + \ldots + y_{v_l} (x_{e_{l-1}}+ x_{e_l}) \\
		  = & ~ (y_{v_1}+y_{v_{2}}) x_{e_1} + (y_{v_2}+y_{v_{3}}) x_{e_2} + \ldots + (y_{v_l}+y_{v_{1}}) x_{e_l} - y_{v_1} \\
		  = & ~ \sum_{i=1}^{l} w_{e_i} x_{e_i} - y_{v_1} \\
		  = & ~ w(T,\boldsymbol{x}) - y_{v_1}.
	\end{split}
\end{equation}
\end{proof}

Since odd degree vertices occur in pairs,
Lemma \ref{thm:odd.vertex.elimination} suggests that we may remove all odd degree vertices in $G'$ by the operations \eqref{eq:operation-integrality} without violating the optimality of $\boldsymbol{x}$.
Lemma \ref{thm:circuit.operation} suggests that we may further remove all even closed trails in $G'$ by the operations \eqref{eq:operation-integrality} without violating the optimality of $\boldsymbol{x}$.
Therefore, we may assume that $G'$ has no odd degree vertices and no even closed trails, which implies that $G'$ consists of vertex disjoint odd cycles.

\subsection{Constructing approximate core allocations from dual optimal solutions}
Recall that $\boldsymbol{x}$ is a half-integral optimal solution to LP \eqref{eq:core-primal} and $\boldsymbol{y}$ is an optimal solution to LP \eqref{eq:core-dual}.
Our mechanism constructs an approximate core allocation of $\Gamma_G$ from $\boldsymbol{y}$ by referring to $\boldsymbol{x}$.
Let $G'=(V',E')$ be the subgraph of $G$ induced by non-integral edges of $\boldsymbol{x}$.
Without loss of generality, we assume that $G'$ consists of vertex disjoint odd cycles.
To construct an approximate core allocation, 
we first round $\boldsymbol{x}$ into an integral feasible solution such that the rounding loss is minimized in every odd cycle,
then define an approximate core allocation from $\boldsymbol{y}$ by equally sharing the rounding loss among players in every odd cycle.
The rounding procedure helps to ensure the validity of our approximate core allocation.
The procedure of minimizing and equally sharing the rounding loss helps to achieve the best factor of our approximate core allocation.

Before proceeding, we introduce some notations for simplicity.
Denote by $\mathcal{C}$ the set of odd cycles in $G'$.
For any odd cycle $C\in \mathcal{C}$, let $y_{C}=\min_{v\in C} y_v$, $b_{C} = \min_{v\in C} b_{v}$, and $l_C$ denote the length of $C$.
Let $l'=\min_{C\in \mathcal{C}} l_{C}$ and $b' = \min_{C\in \mathcal{C}} b_C$.
In the following,
we summarize our mechanism as Mechanism \ref{alg:approx.core} and prove the approximation ratio of our mechanism in Lemma \ref{thm:approx.ratio} and Theorem \ref{thm:main-parametric}.

\begin{figure}
	\begin{wbox}
		\begin{alg}
		\label{alg:approx.core}
		{\bf (Approximate Core Allocation Mechanism)}\\
		\begin{enumerate}
			\item Compute $\boldsymbol{x}$ and $\boldsymbol{y}$, optimal solutions to LPs \eqref{eq:core-primal} and \eqref{eq:core-dual},  where $\boldsymbol{x}$ is half-integral.
			\item Modify $\boldsymbol{x}$ so that all non-integral edges form vertex disjoint odd cycles.
			\item $\forall i \in V$, 
			\begin{equation*}\label{eq:mech}
    				a_i =
    					\begin{cases*}
      						{(1 -\frac{1}{l_{C} \cdot b_{i}}) \cdot b_{i} y_{i}} & if $i$ is in a non-integral odd cycle  $C$,  \\
      						{b_{i} y_{i}}				  & otherwise.
    					\end{cases*}
  			\end{equation*}
		\end{enumerate} 
		Output $\boldsymbol{a}=(a_i)_{i\in V}$.
		\end{alg}
	\end{wbox}
\end{figure}

\begin{lemma}\label{thm:approx.ratio}
The allocation computed by Mechanism \ref{alg:approx.core} lies in the $(1 - \frac{1}{l' \cdot b'})$-approximate core of $\Gamma_G=(N,\gamma)$. 
\end{lemma}

\begin{proof}
Let $\boldsymbol{x}'$ be a $b$-matching obtained from $\boldsymbol{x}$ by performing the operations \eqref{eq:operation-integrality} along every non-integral odd cycle $C\in \mathcal{C}$ from a vertex $v$ on $C$ with $y_{v}= y_{C}$.
For the grand coalition $N$, we have
\begin{equation}\label{eq:c21}
	\begin{split}
		  \gamma (N)
		  \geq & ~ w(G,\boldsymbol{x}')\\
		  = & ~ w(G,\boldsymbol{x})-\sum_{C\in \mathcal{C}} y_C\\
		  = & ~ b(G,\boldsymbol{y}) - \sum_{C\in \mathcal{C}} y_C\\
		  = & ~ \sum_{i\in V'} b_{i} y_{i} + \sum_{i\in V\backslash V'} b_{i} y_{i} - \sum_{C\in \mathcal{C}} y_C\\
		  = & ~ \sum_{C\in \mathcal{C}} \sum_{i\in C} ( b_{i} y_{i} - \frac{y_{C}}{l_{C}} )+ \sum_{i\in V\backslash V'} b_{i} y_{i}\\
		  = & ~ \sum_{C\in \mathcal{C}} \sum_{i\in C} ( 1 -  \frac{y_{C}}{y_i} \cdot \frac{1}{l_{C} \cdot b_{i}} ) b_{i} y_{i} + \sum_{i\in V\backslash V'} b_{i} y_{i}\\
		  \geq & ~ \sum_{i\in N} a_{i}.
	\end{split}
\end{equation}
For any coalition $S\subseteq N$ and any maximum $b$-matching $\boldsymbol{z}$ in the induced subgraph $G[S]$, we have
\begin{equation}
\label{eq:best.ratio}
\gamma (S)=\sum_{ij\in G[S]} w_{ij} z_{ij} \leq \sum_{ij\in G[S]} (y_i + y_j) z_{ij}\leq \sum_{i\in S} y_i b_i \leq \frac{1}{1-\frac{1}{l' \cdot b'}}\sum_{i\in S} a_i.
\end{equation}
\end{proof}

Notice that $l'$ and $b'$ in Lemma \ref{thm:approx.ratio} depend on the initial half-integral optimal solution $\boldsymbol{x}$ and odd cycles induced by $\boldsymbol{x}$.
We reformulate Lemma \ref{thm:approx.ratio} as follows.

\begin{theorem}\label{thm:main-parametric}
The allocation computed by Mechanism \ref{alg:approx.core} lies in the $(1 - \frac{1}{l_{\min} \cdot b_{\min}})$-approximate core of the multiple partners matching game on $G$,
where $l_{\min}$ is the minimum length of odd cycles and $b_{\min}$ is the minimum capacity of vertices.
\end{theorem}

Mechanism \ref{alg:approx.core} can be modified by arbitrarily sharing the rounding loss among players in every odd cycle, which also yields valid approximate core allocations.
However, according to the last inequality in \eqref{eq:best.ratio}, the approximation ratio is better than any other case when the rounding loss is shared equally among players in every odd cycle.
In the following, we provide a tight example that suggests that the approximation ratio in Theorem \ref{thm:main-parametric} cannot be improved relative to the underlying LP-relaxation.

\begin{example}\label{example}
Consider the infinite family of graphs defined as follows.
Let $G_{n}$ be a graph consisting of $2n$ vertex disjoint odd cycles of length $l$, where $n$ is an integer.
Every vertex shares the same odd integral capacity $b$,
and every edge shares the same unit weight.
It is easy to verify that $\boldsymbol{x}$ with $x_e  = b/2$ for every edge $e$ is an optimal primal solution and $\boldsymbol{y}$ with $y_v = 1/2$ for every vertex $v$ is an optimal dual solution to the natural LP-relaxation of the maximum $b$-matching problem on $G_n$.
The maximum weight of $b$-matchings in $G_n$ is 
\begin{equation*}
\frac{l b-1}{2}\cdot 2n=n (l b -1).
\end{equation*}
Meanwhile, Mechanism \ref{alg:approx.core} allocates
\begin{equation*}
(1-\frac{1}{l b}) \cdot b \cdot \frac{1}{2}=\frac{l b-1}{2 l}
\end{equation*}
to each vertex in $G_n$,
the sum of which is precisely the maximum weight of $b$-matchings in $G_n$.
In case a connected graph is required, we may choose one vertex from each odd cycle and connect every pair of chosen vertices with an edge of weight $\epsilon$, where $\epsilon$ tends to zero.
\end{example}

The \emph{integrality gap} of a problem is the worst-case ratio over all instances of the optimal value of the problem to the optimal value of the natural LP-relaxation.

\begin{theorem}\label{thm:main-integrality.gap}
The integrality gap of the maximum $b$-matching problem relative to the natural LP-relaxation is $\frac{2}{3}$.
\end{theorem}
\begin{proof}
Let $l_{\min}$ be the minimum length of odd cycles and $b_{\min}$ be the minimum vertex capacity.
Let $\boldsymbol{x}^*$ be a maximum $b$-matching and $\boldsymbol{a}$ be an allocation computed by Mechanism \ref{alg:approx.core}.
Then
\begin{equation}
w(G,\boldsymbol{x}^*) \geq \sum_{i\in N} a_{i} \geq (1 - \frac{1}{l_{\min} \cdot b_{\min}}) \cdot b(G,\boldsymbol{y}) = (1 - \frac{1}{l_{\min} \cdot b_{\min}}) \cdot w(G,\boldsymbol{x}),
\end{equation}
implying $\frac{w(G,\boldsymbol{x}^*)}{w(G,\boldsymbol{x})} \geq \frac{2}{3}$.
This provides a lower bound for the integrality gap of LP-relaxation \eqref{eq:core-primal}.

An upper bound on the integrality gap of LP-relaxation \eqref{eq:core-primal} can be placed by setting $l=3$ and $b=1$ in Example \ref{example}.
In this case, the maximum weight of $b$-matchings is $2n$ and the optimal value to LP-relaxation \eqref{eq:core-primal} is $3n$.
Thus there are infinitely many graphs achieving $\frac{2}{3}$.
\end{proof}

Theorem \ref{thm:main-parametric} guarantees that Mechanism \ref{alg:approx.core} always computes an allocation in the $\frac{2}{3}$-approximate core for any multiple partners matching game.
Theorem \ref{thm:main-integrality.gap} implies that the factor cannot be better than $\frac{2}{3}$,
since the factor is obtained by comparing the value of a primal solution against the value of a dual solution, which can never be better than the integrality gap.

\begin{theorem}\label{thm:main-intro}
The $\frac{2}{3}$-approximate core of multiple partners matching games is always non-empty.
Moreover, $\frac{2}{3}$ is the best possible factor relative to the underlying LP-relaxation.
\end{theorem}

\section{Concluding remarks}

As Tutte \cite{Tutt54} observed, there is a natural reduction from $b$-matchings to matchings by duplicating every vertex $v$ into $b_v$ copies and every edge $uv$ into $b_u b_v$ copies.
Since every coalition in a multiple partners matching game corresponds to a coalition in the resulting matching game with the same value,
the reduction implies that an approximate core allocation in the latter is also an approximate core allocation in the former.
The reduction also suggests that a mechanism for approximate core allocations can be generalized from matching games to multiple partners matching games and achieve the same approximate ratio in both settings.
However, the reduction causes an exponential increase in the instance size.
Consequently, the mechanism for matching games has a slow pseudo-polynomial runtime on multiple partners matching games.
In contrast, Mechanism \ref{alg:approx.core} is polynomial in the size of $b$-values and has a better approximation ratio.
Besides, we provide a higher point of view for interpreting the duality principle and the rounding technique behind the LP-based mechanism.


\bibliographystyle{habbrv}
\bibliography{reference}
\end{document}